\documentclass{article}
\pdfoutput=1
\usepackage[utf8]{inputenc}
\usepackage{microtype}
\usepackage{enumerate,ifthen,xspace}
\usepackage{graphicx,color}
\usepackage{amsmath,amsfonts,amstext,amssymb}
\makeatletter\@ifclassloaded{llncs}{}{\usepackage{amsthm}}\makeatother
\usepackage{subfig}
\usepackage[numbers]{natbib}
\usepackage{circle}
\usepackage{aliascnt}
\usepackage[pdfusetitle]{hyperref} %hyperref also incompatible with llncs?
\usepackage{cleveref} % newer version of cleveref can do \cref{footnote:...}, but screws up refs to \tag{...}
\usepackage{tikz}

\ifx\pgfversion\undefined
  \input{pgfexternal.tex}
\else

  \usetikzlibrary{trees}
  \usetikzlibrary{backgrounds}
  \usetikzlibrary{scopes}
  \usetikzlibrary{decorations.pathmorphing}

  \tikzstyle{model}=[temprel/.style={->,draw},
                     event/.style={fill=white,midway,font=\scriptsize,inner sep=1pt},
                     eventrel/.style={temprel,every node/.style=event},
                     accrel/.style={->,draw,opacity=0.4,thick},
                     mapping/.style={->,draw,dashed,opacity=0.2,semithick,shorten <=3pt,shorten >=3pt},
                     lr/.style={bend left},
                     rl/.style={bend right},
                     labeled/.style={fill=none,rectangle,font=\small,outer sep=2pt, inner sep=0pt},
                     unlabeled/.style={labeled,inner sep=1.5pt,circle,fill}]
  \tikzstyle{ettree}=[model,
                      level/.style={level distance=1.3cm,sibling distance=1.3cm},
                      edge from parent/.style={temprel},
                      edge from parent/.style={eventrel},
                      grow'=up]
  \tikzstyle{etltree}=[ettree,
                       every child node/.style={unlabeled}]
  \tikzstyle{deletltree}=[ettree,
                          every child node/.style={labeled}]
  \tikzstyle{delmodel}=[model,
                        every node/.style={labeled}]
                       
  \newcommand{\mkinfset}[1]{
    \begin{pgfonlayer}{background}
      \filldraw[model,accrel,-,line width=23pt,line join=round] #1 -- cycle;
      \filldraw[white,line width=22pt,line join=round] #1 -- cycle;
    \end{pgfonlayer}
  }

  \ifthenelse{\lengthtest{\pgfversion cm<1.18cm}}{
  }{
    \ifx\regeneratepgf\undefined\pgfrealjobname{\jobname}
    \else\pgfrealjobname{\regeneratepgf}
    \fi
  }

\fi

\author{Andreas Witzel}
\title{Characterizing perfect recall\\using next-step temporal operators\\in S5 and sub-S5 Epistemic Temporal Logic}

\newcommand{\PR}{\ensuremath{\mathsf{PR}}\xspace}
\newcommand{\PRs}{\ensuremath{\mathsf{sPR}}\xspace}
\newcommand{\PRsprime}{\ensuremath{\mathsf{wsPR}}\xspace}

\newcommand{\PRhc}{\ensuremath{\mathsf{PR_{hc}}}\xspace}
\newcommand{\PRhcl}{\ensuremath{\mathsf{PR_{hc}^\ell}}\xspace}
\newcommand{\PRee}{\ensuremath{\mathsf{PR_{ee}}}\xspace}
\def\dfn{\textbf}
\def\dfnless{\dfn} %{\emph}
\def\F{\ensuremath{\mathcal{F}}\xspace}
\newcommand{\FSfive}{\dot\F}
\newcommand{\acc}{\sim}
\DeclareMathOperator{\accSfive}{\dot\acc}
\newcommand{\nextstep}{\Diamond}
\newcommand{\after}[1]{\langle#1\rangle}
\newcommand{\nextstepall}{\Box}
\newcommand{\afterall}[1]{[#1]}
\renewcommand{\implies}{\rightarrow}
\renewcommand{\phi}{\varphi}

\newcommand{\tiff}{if and only if\xspace}
\newcommand{\mpunct}{\enspace}
\DeclareMathOperator{\EE}{EE}

\ifx\crefformat\undefined
\newcommand{\crefformats}[7]{}
\else
\newcommand{\crefformats}[7]{%
  \crefformat{#3}{##2\ifthenelse{\equal{#4}{}}{}{#4~}#1##1#2##3}
  \Crefformat{#3}{##2\ifthenelse{\equal{#5}{}}{}{#5~}#1##1#2##3}
  \crefrangeformat{#3}{\ifthenelse{\equal{#6}{}}{}{#6~}##3#1##1#2##4--##5#1##2#2##6}
  \Crefrangeformat{#3}{\ifthenelse{\equal{#7}{}}{}{#7~}##3#1##1#2##4--##5#1##2#2##6}
  \crefmultiformat{#3}{\ifthenelse{\equal{#6}{}}{}{#6~}##2#1##1#2##3}{ and~##2#1##1#2##3}{, ##2#1##1#2##3}{ and~##2#1##1#2##3}
  \Crefmultiformat{#3}{\ifthenelse{\equal{#7}{}}{}{#7~}##2#1##1#2##3}{ and~##2#1##1#2##3}{, ##2#1##1#2##3}{ and~##2#1##1#2##3}
  \crefrangemultiformat{#3}{\ifthenelse{\equal{#6}{}}{}{#6~}##2#1##1#2##3}{ and~##2#1##1#2##3}{, ##2#1##1#2##3}{ and~##2#1##1#2##3}
  \Crefrangemultiformat{#3}{\ifthenelse{\equal{#7}{}}{}{#7~}##2#1##1#2##3}{ and~##2#1##1#2##3}{, ##2#1##1#2##3}{ and~##2#1##1#2##3}
}

\makeatletter
\@ifclassloaded{llncs}{
  \crefformats{}{}{section}{Sect.}{Section}{Sects.}{Sections}
  \crefformats{}{}{appendix}{App.}{Appendix}{App.}{Appendices}
  \crefformats{}{}{subsection}{Sect.}{Section}{Sects.}{Sections}
  \crefformats{}{}{subsubsection}{Sect.}{Section}{Sects.}{Sections}
  \crefformats{}{}{chapter}{Chap.}{Chapter}{Chaps.}{Chapters}
  \crefformats{}{}{part}{Part}{Part}{Parts}{Parts}
  \crefformats{}{}{figure}{Fig.}{Figure}{Figs.}{Figures}
  \crefformats{}{}{table}{Table}{Table}{Tables}{Tables}
}{
   \crefformats{}{}{section}{Section}{Section}{Sections}{Sections}
   \crefformats{}{}{appendix}{Appendix}{Appendix}{Appendices}{Appendices}
   \crefformats{}{}{subsection}{Section}{Section}{Sections}{Sections}
   \crefformats{}{}{subsubsection}{Section}{Section}{Sections}{Sections}
   \crefformats{}{}{chapter}{Chapter}{Chapter}{Chapters}{Chapters}
   \crefformats{}{}{part}{Part}{Part}{Parts}{Parts}
   \crefformats{}{}{figure}{Figure}{Figure}{Figures}{Figures}
   \crefformats{}{}{table}{Table}{Table}{Tables}{Tables}
}
\makeatother
\crefformats{}{}{footnote}{footnote}{footnote}{footnotes}{footnotes}
\fi

\ifx\finalversion\undefined
\newcommand{\marginlabel}[2]{%
  \mbox{}%
  \marginpar[\raggedleft\hspace{0pt}#1]{\raggedright\hspace{0pt}#2}%
}
\ifx\tikzstyle\undefined
\newcommand{\todoar}[2][]{\todo[#1]{#2}}
\else\tikzstyle{todoarrow}=[opacity=0.4,gray,-stealth]
\newcommand{\todoar}[2][]{%
  \marginlabel{\small #2}%\tikz[remember picture,overlay]\node(todoarrowstart){};}%
        {\tikz[remember picture,overlay,baseline=(todoarrowstart.220)]\node(todoarrowstart){};$\lhd$ \small #2}%
  \ifthenelse{\equal{#1}{}}{}{{\color{red}[}#1{\color{red}]}}%
  \tikz[remember picture,overlay]\node[inner sep=2pt](todoarrowend){};%
  \tikz[remember picture,overlay]\path(todoarrowstart)edge[todoarrow,out=190,in=-45](todoarrowend);%
}
\fi
\newcommand{\todo}[2][]{%
  \marginlabel{$\rhd$ {\small #2}}{$\lhd$ \small #2}%
  \ifx\color\undefined%
  \ifthenelse{\equal{#1}{}}{}{{[}#1{]}}%
  \else%
  \ifthenelse{\equal{#1}{}}{}{{\color{red}[}#1{\color{red}]}}%
  \fi%
}
\fi

\newcounter{autoexternalpgf}
\newcommand{\suchthat}{\,|\,}

\makeatletter
\@ifclassloaded{beamer}{}{%
  \newcommand{\newtheoremwithalias}[3]{%
    \ifx\newaliascnt\undefined
    \newcounter{#1}
    \else
    \newaliascnt{#1}{#2}
    \fi
    \newtheorem{#1}[#1]{#3}
    \ifx\aliascntresetthe\undefined\else
    \aliascntresetthe{#1}
    \fi
  }
  \ifx\creflastconjunction\undefined\else
  \renewcommand{\creflastconjunction}{ and }
  \fi

  \@ifclassloaded{llncs}{
    \if@envcntsame\errmessage{cleveref naming doesn't work because no aliascntrs used in llncs.cls}
    \else\if@envcntsect
    \newtheorem{observation}{Observation}[section]
    \newtheorem{fact}{Fact}[section]
    \else
    \newtheorem{observation}{Observation}
    \newtheorem{fact}{Fact}
    \fi\fi
    \newcommand{\qedhere}{\qed}
    \crefformats{(}{)}{equation}{}{Equation}{}{Equations}
  }{
    \newtheorem{theorem}{Theorem}[section]
    \newtheoremwithalias{lemma}{theorem}{Lemma}
    \newtheoremwithalias{proposition}{theorem}{Proposition}
    \newtheoremwithalias{corollary}{theorem}{Corollary}
    \newtheoremwithalias{observation}{theorem}{Observation}
    \newtheoremwithalias{fact}{theorem}{Fact}
    \newtheoremwithalias{conjecture}{theorem}{Conjecture}

    \ifx\theoremstyle\undefined\else\theoremstyle{remark}\fi
    \newtheoremwithalias{numberednote}{theorem}{Note}
    \newtheoremwithalias{numberedremark}{theorem}{Remark}

    \ifx\theoremstyle\undefined\else\theoremstyle{definition}\fi
    \newtheoremwithalias{definition}{theorem}{Definition}
    \newtheoremwithalias{example}{theorem}{Example}

    \crefformats{(}{)}{equation}{Equation}{Equation}{Equations}{Equations}

  }

  \crefformats{}{}{theorem}{Theorem}{Theorem}{Theorems}{Theorems}
  \crefformats{}{}{lemma}{Lemma}{Lemma}{Lemmas}{Lemmas}
  \crefformats{}{}{proposition}{Proposition}{Proposition}{Propositions}{Propositions}
  \crefformats{}{}{corollary}{Corollary}{Corollary}{Corollaries}{Corollaries}
  \crefformats{}{}{observation}{Observation}{Observation}{Observations}{Observations}
  \crefformats{}{}{fact}{Fact}{Fact}{Facts}{Facts}
  \crefformats{}{}{conjecture}{Conjecture}{Conjecture}{Conjectures}{Conjectures}
  \crefformats{}{}{numberednote}{Note}{Note}{Notes}{Notes}
  \crefformats{}{}{numberedremark}{Remark}{Remark}{Remarks}{Remarks}
  \crefformats{}{}{definition}{Definition}{Definition}{Definitions}{Definitions}
  \crefformats{}{}{example}{Example}{Example}{Examples}{Examples}
}
\makeatother

\begin{document}

\maketitle

\begin{abstract}
  We review the notion of perfect recall
  in the literature on interpreted systems, game theory, and epistemic logic.
  In the context of Epistemic Temporal Logic (ETL),
  we give a (to our knowledge) novel frame condition for perfect recall,
  which is local and can straightforwardly be translated to a defining formula
  in a language that only has next-step temporal operators.
  This frame condition also gives rise to a complete axiomatization for S5 ETL frames with perfect recall.
  We then consider how to extend and consolidate the notion of perfect recall in sub-S5 settings,
  where the various notions discussed are no longer equivalent.
\end{abstract}

\section{Introduction}
\label{sec:introduction}

Perfect recall is an epistemic-temporal notion concerning an agent's ability to remember the past.
It does not entail that all knowledge an agent currently has is preserved---%
in fact, certain (negative) knowledge must be lost as the agent learns.
For example ``I know that I don't know $p$'' is lost when I learn~$p$.
Rather, it means that an agent remembers all the information he once \emph{had},
and can use it to reason about the present.
%or vice versa, that his \emph{uncertainty} can only decrease over time.
%in senses which we formalize in \cref{sec:defin-perf-recall}.
% the following sense:
% If an agent was able to distinguish two situations at some point in their histories,
% then he is able to distinguish them now.
% Put differently, if the agent's ``epistemic experience'' differs in any two histories,
% then he is able to distinguish the resulting situations.

The notion of perfect recall is well-studied in game theory
(see, e.g., \cite{piccione_absent-minded_1997} and~\cite[Section~11.1.3]{osborne_course_1994})
and in the distributed computing/interpreted systems/temporal logic literature~%
\cite{fagin_reasoning_1995,van_der_meyden_axioms_1993,van_der_meyden_complete_2003,halpern_complete_2004,parikh_knowledge_2003}
(see \cite{halpern_ambiguities_1997} for a discussion from the viewpoint of the intersection of both).
In recent years, it has also been discussed in the epistemic logic community~%
\cite{van_benthem_games_2001,van_benthem_tree_2006,van_benthem_merging_2009,isaac_synchronizing_2010},
specifically in the context of Epistemic Temporal Logic (ETL) (see, e.g., \cite{van_benthem_tree_2006}).
ETL is an epistemic logic (or rather, family of logics) with added event modalities, interpreted on tree models,
and it is intended to capture interactions of agents over time.
Unlike with interpreted systems, the motivation is not dominated by the idea
of processes with control programs which govern their behavior into infinity,
and therefore, ETL does not necessarily aim at making statements about such long-term behavior.
Although the concept of histories (corresponding to runs in interpreted systems) does exist
and variants of ETL with long-range operators such as ``until'' have been studied,
common ETL languages focus on the local perspective and lack global temporal operators
that can talk about indefinite time spans.
In the following, we use ``ETL'' in the narrow sense of epistemic temporal logic with only next-step (local) temporal modalities.

Due to the differences between the mentioned frameworks,
various definitions and conditions for perfect recall exist.
Our aim here is to review and examine the relevant ones
from the viewpoint of a logic without long-range temporal operators.
We use the history-based framework and notation of ETL (also used by \citet{parikh_knowledge_2003}),
but interpreted systems are for our purposes essentially equivalent \cite{pacuit_comments_2007}
and so our considerations should be transferable, mutatis mutandis.

\medskip

In synchronous settings (i.e., where agents have access to a global clock),
there is a well-known characterization of perfect recall,
including a simple frame condition and a definition that can be expressed in ETL.
In contrast, without synchronicity the existing characterizations of perfect recall
seem to be inherently global in the sense that they talk about whole (prefixes of) histories (or runs),
or refer to some arbitrary point in them.
This precludes a direct translation into a first-order frame condition
(since reachability is not definable in first-order logic)
and also makes it unclear whether the notion can be expressed in ETL.
This is not an issue for the interpreted systems community,
since they by default have long-range temporal operators available.%
\footnote{It is not an issue for the game theory community because they do not commonly concern themselves with formal languages.}
Indeed, \citet{van_der_meyden_axioms_1993} axiomatized perfect recall using the ``until'' operator.
However, this does not transfer to logics which only have next-step temporal operators.

\medskip

After briefly introducing ETL in \cref{sec:epist-temp-logic}
and discussing perfect recall from the perspective of ETL in \cref{sec:defin-perf-recall},
we therefore in \cref{sec:char-pras-etl} propose an alternative characterization of perfect recall
which only uses single steps of temporal (and epistemic) succession.
We show that it is first-order definable and definable in ETL,
and we give a complete axiomatization for S5 ETL with perfect recall.
We then explore sub-S5 settings in \cref{sec:sub-s5-settings}.
\Cref{sec:conclusions} concludes.

\section{Epistemic-temporal logic (ETL)}
\label{sec:epist-temp-logic}

We focus on the single-agent case since perfect recall is a property inherent to one agent;
all our considerations carry over to the multi-agent case.
We consider models over some finite set~$E$ of \dfn{events}.
A \dfn{history} $h\in E^*$ is a finite sequence of events,
and we denote the empty history by~$\epsilon$.
We denote \dfn{sequences} simply by listing their elements,
possibly preceded by a prefix sequence.
For two histories $h,h'$ we write $h\leadsto_eh'$ if $h'=he$,
that is, if~$h'$ \dfn{extends}~$h$ by one event~$e$.
We write $h\leadsto h'$ if $h\leadsto_eh'$ for some event $e$.
We denote the transitive and reflexive closure of~$\leadsto$ by $\leadsto^*$,
so $h\leadsto^*h'$ says that~$h$ is a \dfn{prefix} of~$h'$ (possibly $h'$ itself),
or vice versa, $h'$ is an \dfn{extension} of~$h$.
For $h\leadsto^*h'$, we sometimes also write $h\preceq h'$,
and $h\prec h'$ for $h\preceq h'$ with $h\neq h'$.
A \dfn{protocol} $H\subseteq E^*$ is a finite set of histories closed under taking prefixes,
intuitively representing the allowed evolutions of the system.

An \dfn{(epistemic) accessibility relation} is a binary relation $\acc\ \subseteq H\times H$ on histories.
It specifies, for any given history~$h$,
the histories that the agent \dfn{considers possible} at~$h$.
Various conditions can be imposed on~$\acc$,
making it capture various notions of knowledge or belief
(for details see, e.g., \cite{chellas_modal_1980}).
Most commonly, the relation is assumed to be an \emph{equivalence relation},
making it capture a notion of ``correct knowledge''.
This case is also referred to as~\dfn{S5}.
We refer to less restrictive cases as \dfn{sub-S5}.
Relevant properties include \emph{transitivity},
reflecting \dfn{positive introspection},
and \emph{Euclideanness},
reflecting \dfn{negative introspection}.
Another property is that of \dfn{synchronicity},
which intuitively means that the agent has access to a clock.
It holds if, whenever $h\acc h'$, then the lengths of~$h$ and~$h'$ are equal.

An \dfn{ETL frame} is a tuple $\F=\langle E,H,\acc\rangle$ consisting of
a set of events~$E$, a protocol~$H$ and an epistemic accessibility relation~$\acc$.
We will usually omit~$E$ and~$H$ for the sake of clarity and implicitly assume that
any events or histories we talk about belong to~$E$ or~$H$, respectively.
An ETL frame can be viewed as a temporal \dfn{tree} (induced by the $\leadsto$ relation)
with epistemic accessibilities between nodes.
We will also consider an extension to \dfn{forests},
which are (finite) sets of trees with distinct roots and possibly interrelating epistemic accessibilities.%
\footnote{The distinction between ETL frames and forests corresponds
  to the \emph{unique initial state} condition in the interpreted systems literature.}
Most of our considerations apply both to trees and to forests,
and only in \cref{sec:sub-s5-settings} we have to make the distinction explicit.
We use properties of the epistemic accessibility relation to specify frames with a corresponding relation;
for example, by an \emph{S5 frame} we mean a frame with an S5 accessibility relation.

The \dfn{language} of ETL consists of a finite set $\mathsf{At}$ of propositional atoms
and of all formulas built from those according to the following grammar:
\[
p~|~\neg\varphi~|~\varphi\wedge\psi~|~K\varphi~|~\after e \varphi\mathpunct,
\]
where $p\in\mathsf{At}$ and $\varphi,\psi$ are formulas.
Intuitively, $K\varphi$ means that the agent knows $\varphi$,
and $\after e\varphi$ means that event $e$ can occur and afterwards $\varphi$ will hold.
The remaining propositional connectives are defined as abbreviations as usual,
and the duals of the modalities are denoted by~$L$ (dual of~$K$) and $\afterall e$ (dual of~$\after e$).
We write $\nextstep\phi$ to abbreviate $\bigvee_{e\in E}\after e\phi$.

A \dfn{valuation} $V:\mathsf{At}\to 2^H$ assigns to each atom the set of histories where it is true.
We write $\F,V,h\models\varphi$ for a frame~\F, a valuation~$V$ and a history~$h$ of the protocol of \F
to say that $\varphi$ is \dfn{satisfied} by~$\F,V,h$.
Satisfaction of formulas is defined inductively as usual,
starting with the truth values of atoms as given by~$V$,
and with the following semantics:
\begin{align*}
  \F,V,h&\models p &&\text{iff $h\in V(p)$}\\
  \F,V,h&\models \neg\varphi &&\text{iff $\F,V,h\not\models\varphi$}\\
  \F,V,h&\models \varphi\wedge\psi &&\text{iff $\F,V,h\models\varphi$ and $\F,V,h\models\psi$}\\
  \F,V,h&\models K\varphi &&\text{iff for each $h'\in H$ with $h\acc h'$: $\F,V,h'\models\varphi$}\\
  \F,V,h&\models \after e\varphi &&\text{iff $he\in H$ and $\F,V,he\models\varphi$}
\end{align*}
A formula~$\varphi$ is said to be \dfn{valid} on $\F,V$ \tiff $\F,V,h\models\varphi$ for all~$h\in H$.
It is said to be valid on $\F$ \tiff it is valid on $\F,V$ for all valuations~$V$.
It is said to \dfn{define} a class~$C$ of frames
\tiff it is valid exactly on the frames in~$C$.

For a binary relation $R$ on histories
we write $[h]_R=\{h'\suchthat hRh'\}$ to denote the \dfn{image} of~$h$ under~$R$.
If $H$ is a set of histories, we write $[H]_R$ for $\bigcup_{h\in H}[h]_R$.
Obviously, if $R$ is an equivalence relation, then $[h]_R$ is the equivalence class of $h$ with respect to~$R$.

\section{Perfect recall}
\label{sec:defin-perf-recall}

In this section, we review existing definitions of perfect recall,
give intuitive justifications for the notions,
and examine how they relate to each other.

Like in distributed systems, and unlike in game theory,
in ETL there is no notion of turns and no notion of agency associated with events.
An event is just an event and comes with no specification as to who performs it.
In game theory, turn-taking typically makes successive situations distinguishable
and perfect recall can often be formulated as ``remembers all his actions''.
In ETL, these features are not part of the model%
\footnote{Put differently,
events may perfectly well be caused externally and go completely unnoticed by any agent.
If we do want to attribute certain events to certain agents
(a view which certainly suggests itself in the single-agent case,
and is also customary in game theory),
% it may seem counter-intuitive to attribute perfect recall
% to an agent who does not even remember whether or not he has performed some action.
% However, as mentioned above, ETL events do in fact not come with any notion of agency (or turns),
% so this poses no conceptual problems:
then corresponding observability conditions for the agent performing a particular action
can be specified separately,
and our definitions of perfect recall will not interfere.}
and so we have to use other ways to express the notion.

Whether or not time is part of the agent's perception
is exactly what makes the difference between asynchronous and synchronous systems.
Synchronicity can be defined separately, if it is desired;
we are interested in perfect recall as an independent property,
which does not interfere with (a)synchronicity,
and thus should not assume or imply that agents perceive time.

\subsection{Basic definitions}
\label{sec:basic-definitions}

We start by giving some intuition about the notions we are going to define.
As mentioned, an agent with perfect recall can at any point
remember all the information that he had at any previous point in time,
and is able to exclude any possibilities for the current state of the world
which are inconsistent with that information.

\paragraph{Two different intuitions for perfect recall}
\label{sec:intuitions}

There are two related ways in which an agent might detect such inconsistencies,
the first on the level of epistemic states,
and the second on the level of the semantic structures that model them.%
\footnote{See \cref{sec:conclusions} for some discussion
related to the question of which aspects of the model an agent can access.}
Consider a perfect-recall agent in some state of the world, in ETL terms a history, $h$,
and some other history $h'$.

Firstly, if in state $h'$ the agent would have gone through a different sequence of epistemic states
than he actually has in $h$,
then he can exclude the possibility of $h'$, since he can recall all his epistemic states.

Secondly, if in the state before $h$ the agent was certain
that the world was \emph{not} in a state along history $h'$,
then at $h$ the agent can exclude the possibility that the world is in state $h'$,
since he can recall his previous assessments.
Put differently, if $h'$ is \emph{not} an extension of some history considered possible before,
then the agent can exclude the possibility of $h'$ since there would have been no way for the world to evolve to~$h$.

We formalize these two intuitions in the following.
As we will see, they are equivalent in the context of S5,
while in the general case, neither implies the other (see \cref{sec:sub-s5-settings}).

\paragraph{Formalizing the intuitions}
\label{sec:form-intu}

The first notion is the one most commonly used as starting point in the literature.
It uses the idea of local-state sequences, or ``epistemic experiences'',
meaning sequences of epistemic states that the agent has gone through.
Repetitions of identical states are ignored,
since the agent has no way of discriminating between two states in which he has the same epistemic state.
As in game theory (cf.~\cite[Section~11.1.3]{osborne_course_1994}),
we identify an agent's epistemic state with his information set,
i.e., the set of accessible worlds.

\begin{definition}
  \label{dfn:pr-ee}
  Given an ETL frame and a history $e_1\dots e_\ell$,
  the agent's \dfn{epistemic experience} is the sequence
  \[
  \EE(e_1\dots e_\ell):=[\epsilon]_{\acc}\ [e_1]_{\acc}\ [e_1e_2]_{\acc}\ \dots\ [e_1\dots e_\ell]_{\acc}
  \]
  of epistemic states (or information sets, in game theory parlance) he has gone through.

  We say that the epistemic experiences in two histories $h, h'$ are \dfn{equivalent modulo stutterings},
  in symbols $\EE(h)\approx\EE(h')$,
  \tiff the sequences with all repetitions of subsequent identical sets removed
  are equivalent.

  An ETL frame has \dfn{perfect recall with respect to epistemic experience (\PRee)} \tiff,
  whenever $h\acc h'$, we have $\EE(h)\approx\EE(h')$.
\end{definition}

The second definition is a slight (but equivalent) variant of a notion
most commonly used in the interpreted systems literature.%
\footnote{See~\cite[p.~204]{halpern_complexity_1989} (who call perfect recall ``no forgetting'')
and~\cite[Proposition~2.1(a)]{van_der_meyden_axioms_1993}.}
In the literature it has been mostly used as a technical condition characterizing \PRee,
but by rephrasing it we can provide it with an independent motivation.
% and use it as notion in its own right.
\begin{definition}
  \label{dfn:pr-hc}
  An ETL frame has \dfn{perfect recall with respect to history consistency (\PRhc)} \tiff
  for any histories~$h,h'$ and event~$e$ with $he\acc h'$,
  there is some history $h''$ with $h\acc h''\leadsto^*h'$
  (i.e., some prefix of $h'$ is epistemically accessible from $h$).

  Put differently, the condition is that for each history $h$ and event $e$, we have
  \[
  [he]_\acc\subseteq[[h]_\acc]_{\leadsto^*}\mpunct.
  \]
\end{definition}
This second formulation suggests an intuitive reading:
A frame has \PRhc if all histories considered possible after some event
are extensions of histories considered possible before the event.

We refine this notion in \cref{sec:char-pras-etl}
to obtain one that is more fine-grained and detects more inconsistencies in sub-S5 settings;
however, in the context of S5 the notions are equivalent,
so for simplicity we stick with this definition for now.

\medskip

There are two further related conditions in the literature,
which we will state next.

\begin{definition}[{cf.~\cite{van_der_meyden_complete_2003,van_benthem_merging_2009}}] %[{cf.~\cite[p.~503, Remark~3]{van_benthem_merging_2009}}]
  An ETL frame has \dfn{synchronous perfect recall (\PRs)} \tiff
  for each $h,h',e$ with $he\acc h'$ there is $h''$ with $h\acc h''\leadsto h'$.
  Put differently,
  the condition is that for each history $h$ and event $e$, we have
  \[
  [he]_\acc\subseteq [[h]_\acc]_\leadsto\mpunct.
  \]
\end{definition}
That is, all histories considered possible after some event
are extensions of histories considered possible before the event \emph{by exactly one event}.

\begin{definition}[{cf.~\cite[p.~503, Definition~11]{van_benthem_merging_2009}}]
  An ETL frame has \dfn{weak synchronous perfect recall (\PRsprime)} \tiff
  for all $h,h',e,e'$ with $he\acc h'e'$, we have $h\acc h'$.
  Put differently, 
  for each history $h$ and event $e$ we have
  \[
  [he]_\acc\subseteq [[h]_\acc]_\leadsto\cup\{\epsilon\}\mpunct.
  \]
\end{definition}

It seems difficult to find an intuitive justification for this last definition:
Why should an agent with perfect recall be characterized to consider possible,
after some event, one-step extensions of histories previously considered possible
\emph{or the empty history}?
In order to get a better understanding of the various notions,
we now take a closer look at how they relate.

\subsection{Relating the notions}
\label{sec:relating-notions}

First of all,
as mentioned above,
the two notions of \PRee and \PRhc are equivalent in S5.
\begin{proposition}
  \label{result:pr-ee-equiv-pr-hc}
  An S5 ETL frame has \PRee \tiff it has \PRhc.
\end{proposition}
\begin{proof}
  This follows immediately from \cite[Proposition~2.1]{van_der_meyden_axioms_1993},
  instantiating what the interpreted systems literature calls ``local states''
  by the set of accessible worlds in ETL frames.

  To give a version of the proof,
  we consider two histories and show that the conditions of \PRee and \PRhc are equivalent.
  We first show that \PRee implies \PRhc.
  For two empty histories this is obvious,
  so w.l.o.g.~we assume that the first history is non-empty.
  So assume that \PRee holds and that $he\acc h'$.
  We have either of two cases:
  \begin{enumerate}[(i)]
  \item $[h]_\acc=[he]_\acc$. But then $h\acc h'$, and \PRhc is satisfied.
  \item $[h]_\acc\neq [he]_\acc$.
    By \PRee, we must have $\EE(he)\approx\EE(h')$,
    and thus there must be $h''\leadsto h'$ with $[h]_\acc=[h'']_\acc$.
    So again, \PRhc is satisfied.
  \end{enumerate}
  To see that \PRhc implies \PRee, we proceed by induction on the sum of the lengths of the two histories.
  The base case with both empty is straightforward.
  For the induction step, assume that \PRhc holds.
  W.l.o.g.~we assume that the first history is non-empty,
  so we consider $he$ and $h'$ for some event~$e$, with $he\acc h'$.
  \PRhc yields one of two cases:
  \begin{enumerate}[(i)]
  \item $h\acc h'$. Then the induction hypothesis yields $\EE(h)\approx\EE(h')$.
    Furthermore, we have $[h]_\acc=[h']_\acc=[he]_\acc$, so $\EE(h)\approx\EE(he)$.
    Taken together, we obtain $\EE(he)\approx\EE(h')$, and \PRee is satisfied.
  \item $h\acc h''$ for some $h''\prec h'$.
    It then follows that $h'$ is non-empty.
    Let $g$ be its direct predecessor, i.e., $g\leadsto h'$.
    From $h'\acc he$, with \PRhc it follows that there is $h'''\preceq he$ such that $g\acc h'''$.
    If $h'''=he$, then \PRee follows analogously as in the previous case.
    Otherwise we have $h'''\preceq h$.
    The induction hypothesis yields $\EE(h)\approx\EE(h'')$
    as well as $\EE(g)\approx\EE(h''')$.
    Since $h''\preceq g$ and $h'''\preceq h$,
    we must have $\EE(h)\approx\EE(g)$.
    Since $h\leadsto he$ and $g\leadsto h'$, together with $[he]_\acc=[h']_\acc$
    we obtain that $\EE(he)\approx\EE(h')$.\qedhere
  \end{enumerate}
\end{proof}

As we see in the following, the remaining two notions, \PRs and \PRsprime,
are similar to these but impact another property,
namely that of \emph{synchronicity}.
\PRee and \PRhc, on the other hand, do not interfere with synchronicity.
In agreement with the interpreted systems and game theory literature,
we therefore use them as fundamental definitions of perfect recall in the context of S5,
and due to their equivalence, we use \dfn{perfect recall (\PR)} to refer to both.

\medskip

In the presence of synchronicity, not surprisingly, all notions are equivalent.

\begin{proposition}
  \label{result:sync-pr-ee-equiv-pr-s}
  On synchronous S5 ETL frames, all of the above notions (\PRee, \PRhc, \PRs, \PRsprime) are equivalent.
\end{proposition}
\begin{proof}
  Straightfoward with \cref{result:pr-ee-equiv-pr-hc};
  see also \cite{halpern_complete_2004,van_der_meyden_complete_2003}.
\end{proof}

It is easy to see that \PRs implies synchronicity,
since $\leadsto$ is well-founded.%
\footnote{One formulation of $\PRs$ is (cf.~\cref{result:pr-s-definable}):
``If the agent knows that at the next time point $p$ holds, then at the next time point he knows that $p$ holds.''
Intuitively this implies that the mere ticking of the clock affects the agent's mental state.}
Therefore, in general on S5 frames, \PRs is equivalent to \PR \emph{plus} synchronicity.
As for \PRsprime, the relationship is a bit less clear-cut.

\begin{proposition}
  On S5 ETL frames, \PRsprime implies \PR, but not vice versa.
  Put differently, \PRsprime is sound but incomplete with respect to \PR.
\end{proposition}
\begin{proof}
  Soundness is straightforward;
  for incompleteness, see \cref{fig:pr-s-prime-pr-hc}\subref{fig:pr-s-prime-pr-hc:incomp}.
\end{proof}

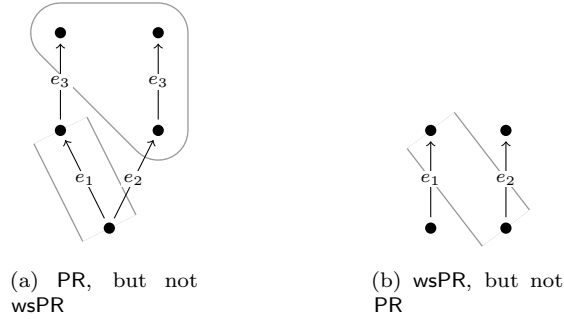
\begin{figure}
  \centering
  \hfill
  \subfloat[\PR, but not \PRsprime]{
    \label{fig:pr-s-prime-pr-hc:incomp}
    \begin{tikzpicture}[etltree]
      \node[unlabeled] (root) {}
        child foreach \c in {1,2} {
          node {} child { node {} edge from parent node {$e_3$} }
          edge from parent node {$e_\c$}
        };
      \mkinfset{(root-1-1.center) -- (root-2-1.center) -- (root-2.center)}
      \mkinfset{(root-1.center) -- (root.center)}
    \end{tikzpicture}
  }
  \hfill
  \subfloat[\PRsprime, but not \PR]{
    ~
    \label{fig:pr-s-prime-pr-hc:nsound}
    \begin{tikzpicture}[etltree]
      \node[unlabeled] (root1) {}
        child { node {} edge from parent node {$e_1$} };
      \node[unlabeled,right of=root1] (root2) {}
        child { node {} edge from parent node {$e_2$} };
      \mkinfset{(root1-1.center) -- (root2.center)}
    \end{tikzpicture}
    ~
  }
  \hfill~
  \caption{Two frames, gray lines indicating information sets.
    \subref{fig:pr-s-prime-pr-hc:incomp} \PRsprime is incomplete with respect to \PR
    (\PRsprime is violated because $e_1e_3\acc e_2e_3$ but $e_1\not\acc e_2$),
    and \subref{fig:pr-s-prime-pr-hc:nsound} \PRsprime is not sound with respect to \PR on ETL forests:
    Intuitively, event $e_1$ lets the agent ``forget'' that he is in the left tree.
    Note that on forests the definition of \PRsprime has to be adjusted by replacing $\{\epsilon\}$
    by the set of all roots.}
  \label{fig:pr-s-prime-pr-hc}
\end{figure}

In that sense, \PRsprime is ``somewhere in between'' \PRs and \PR:
It is strictly implied\footnote{By this we mean ``implies, but is not equivalent to''.}
by \PRs and strictly implies \PR.
What distinguishes \PRsprime is that,
while it (unlike \PRs) does not \emph{imply} synchronicity,
it does (unlike \PR) \emph{presuppose} it,
in that it fails to classify asynchronous frames correctly---%
given that, as we have argued, \PR captures the intuition of perfect recall as an independent property.
We therefore neglect \PRsprime in the further discussion:
Synchronous perfect recall frames are naturally captured by \PRs
and asynchronous ones by \PR,
while \PRsprime has no clear domain of application in our context.

\begin{numberedremark}
  Except where noted, all our considerations carry over to ETL forests
  (sets of ETL frames with possibly interrelating indistinguishabilities).
  The first such note is the fact that on forests,
  \PRsprime is not only incomplete with respect to \PR but also not sound,
  as shown in \cref{fig:pr-s-prime-pr-hc}\subref{fig:pr-s-prime-pr-hc:nsound}.
\end{numberedremark}

Note that the ``local histories'' of \citet{parikh_knowledge_2003}
give rise to frames in which exactly those histories are indistinguishable
in which the epistemic experiences are equivalent modulo stutterings,
so that their frames inherently satisfy \PR.%
\footnote{In fact, the frames they obtain are exactly those S5 frames
  which satisfy \PR and a ``fixed observability'' condition
  stating that, for each agent, there is a fixed set of events
  whose occurrence he is able to distinguish from the (non-)occurrence of others.}
The version used by \citet{pacuit_logic_2006} analogously produces frames
which inherently satisfy \PRs (and thus synchronicity).

\medskip

Finally, we note that \PRs is definable in ETL in a neat way, which also gives rise to an axiomatization.

\begin{proposition}
\label{result:pr-s-definable}
  An S5 ETL frame has \PRs \tiff it validates the following formula:
  \[
  \nextstep L p\implies L \nextstep p
  \]
  Together with S5, this axiomatizes ETL on synchronous S5 frames with perfect recall.
\end{proposition}
This result is well-known and is analogous to
\cite[Theorem~4]{van_der_meyden_complete_2003} and \cite[Theorem~3.6]{halpern_complete_2004},
since \PRs characterizes synchronous frames with perfect recall.%
\footnote{Compare also the Cross Axiom of \citet{dabrowski_topological_1996}.}
It can be proved using \citet{sahlqvist_completeness_1975},
similar to the proof of \cref{result:pr-hcl-definable}.

\medskip

A corresponding characterization of \PR has been given by \citet{van_der_meyden_axioms_1993},
but that characterization uses the ``until'' operator.
To our knowledge, a version using only ``short-sighted'' next-step temporal modalities
has not been discussed.
Providing it is the aim of the next section.

\section{Defining perfect recall in ETL}
\label{sec:char-pras-etl}

We start by giving a ``local'' version of perfect recall with respect to history consistency,
whose intuitive interpretation gives a more fine-grained account of the agent's possibilities for reasoning
than \PRhc does.
As we will see in \cref{sec:sub-s5-settings}, in general the notions are not equivalent,
and the fine-grainedness of the version we propose indeed makes a difference.
In the current section, however, we show that in the context of S5
this definition is equivalent to the preceding ones,
and we exploit its locality to find an axiomatization of perfect recall
using only next-step temporal modalities.

\begin{definition}
\label{dfn:pr-hcl}
  An ETL frame has \dfn{\PRhc, local version (\PRhcl)} \tiff for each history $h$ and event $e$, we have
  \[
  [he]_\acc\subseteq [h]_\acc\cup [[h]_\acc]_\leadsto\cup [[he]_\acc]_\leadsto
  \]
  Put differently, for each $h,h',e$ with $he\acc h'$, either of the following holds:
  \begin{enumerate}[(i)]
  \item\label{dfn:pr-hcl:1} $h\acc h'$
  \item\label{dfn:pr-hcl:2} $h\acc h''\leadsto h'$ for some $h''$
  \item\label{dfn:pr-hcl:3} $he\acc h''\leadsto h'$ for some $h''$.
  \end{enumerate}
\end{definition}

\begin{figure}
  \renewcommand{\thesubfigure}{\relax}
  \captionsetup[subfloat]{labelformat=simple}
  \centering
  \hfill
  \subfloat[if $he\acc h'$, then]{
    \hspace{5mm}
    \begin{tikzpicture}[etltree]
      \node[labeled] (root1) {$h$}
        child {
          node[labeled] {$he$}
          edge from parent node {$e$}
        };
      \node[labeled,right of=root1-1] (root2) {$h'$};
      \draw[accrel] (root1-1) edge[lr] (root2);
    \end{tikzpicture}
    \hspace{5mm}
  }
  \hfill
  \subfloat[(i) $h\acc h'$, or]{
    \hspace{5mm}
    \begin{tikzpicture}[etltree]
      \node[labeled] (root1) {$h$}
        child {
          node[labeled] {$he$}
          edge from parent node {$e$}
        };
      \node[labeled,right of=root1-1] (root2) {$h'$};
      \draw[accrel] (root1-1) edge[lr] (root2)
                    (root1) edge[lr,dashed] (root2.220);
    \end{tikzpicture}
    \hspace{5mm}
  }
  \hfill
  \subfloat[(ii) $h\acc h''\leadsto h'$, or]{
    \hspace{5mm}
    \begin{tikzpicture}[etltree]
      \node[labeled] (root1) {$h$}
        child {
          node[labeled] {$he$}
          edge from parent node {$e$}
        };
      \node[labeled,right of=root1] (root2) {$h''$}
        child {
          node[labeled] {$h'$}
          edge from parent
        };
      \draw[accrel] (root1-1) edge[lr] (root2-1)
                    (root1) edge[lr,dashed] (root2);
    \end{tikzpicture}
    \hspace{5mm}
  }
  \hfill
  \subfloat[(iii) $he\acc h''\leadsto h'$]{
    \hspace{5mm}
    \begin{tikzpicture}[etltree]
      \node[labeled] (root1) {$h$}
        child {
          node[labeled] {$he$}
          edge from parent node {$e$}
        };
      \node[labeled,right of=root1] (root2) {$h''$}
        child {
          node[labeled] {$h'$}
          edge from parent
        };
      \draw[accrel] (root1-1) edge[lr] (root2-1)
                    (root1-1) edge[lr,dashed] (root2.120);
    \end{tikzpicture}
    \hspace{5mm}
  }
  \hfill~
  \caption{Illustrating \PRhcl. Gray, bent arrows indicate accessibilities, showing the directed versions for the sake of illustration.}
  \label{fig:prhcl}
\end{figure}
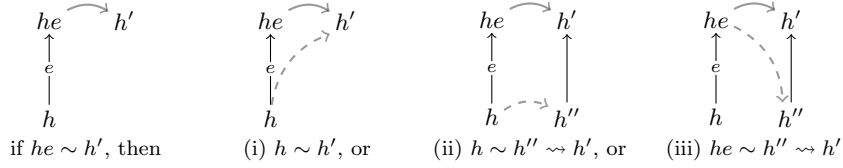

See \cref{fig:prhcl} for an illustration of this definition.
Let us walk through it and compare it with \PRhc.
We consider the history $he$ and refer to $h$ as ``before~$e$'' and to~$he$ as ``after~$e$''.
The definition says that any history considered possible after~$e$
either
\begin{enumerate}[(i)]
\item was considered possible already before~$e$,
  i.e., the agent didn't notice~$e$, nor time passing; or
\item is an extension by one event of a history considered possible before~$e$,
  i.e., the agent correctly thinks one event occurred, though he may not be certain which one; or
\item is an extension by one event of another history considered possible after~$e$.
\end{enumerate}
As illustrated in \cref{fig:bridging-hc-hcl}, this last condition inductively bridges the gap to \PRhc,
and allows the agent to consider possible that several events occur while really just~$e$ is happening.
The difference, as compared to \PRhc, is that now there is a stricter consistency requirement.
If the agent considers possible that several events have happened,
he must obviously be unable to detect some of them, since really just one event happened.
Given that, he must also consider possible the intermediate histories along these several events---%
either from the history after~$e$ or from before~$e$.
Exactly this consistency requirement is inductively captured by the last condition
(in interplay with the second condition).

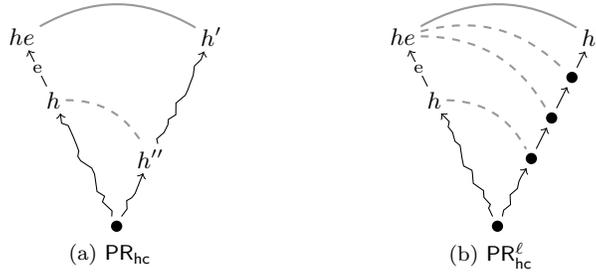
\begin{figure}
  \centering
  \hfill
  \subfloat[\PRhc]{
   \pgfmathsetseed{55}
   \begin{tikzpicture}[etltree,scale=.6]
      \node[unlabeled] (root) {};
      \path (root) ++(-1.4cm,2.8cm) node[labeled] (h) {$h$}
                   ++(-.7cm,1.4cm) node[labeled] (he) {$he$};
      \path (root) ++(.75cm,1.5cm) node[labeled] (h'') {$h''$}
                   ++(1.35cm,2.7cm) node[labeled] (h') {$h'$};

      \draw[eventrel] (h) -- node {e} (he);
      { [eventrel,decoration={random steps,segment length=1mm,amplitude=.5mm}]
        \draw[decorate] (root) -- (h); %node {h} (h);
        \draw[decorate] (root) -- (h''); %node {h''} (h'');
        \draw[decorate] (h'') -- (h');
      }
%      \draw[decorate,decoration={brace,mirror,raise=8pt}] (root) -- node[right,event,xshift=4mm,yshift=-2mm] {$h'$} (h');
      % \node[labeled,right,inner sep=2pt] at (h') {$h'$};
      % \node[labeled,right,inner sep=2pt] at (h'') {$h''$};
      % \node[labeled,left,inner sep=2pt] at (h) {$h$};
      % \node[labeled,left,inner sep=2pt] at (he) {$he$};

      \path (he) edge[accrel,lr,-] (h');
      \path (h) edge[accrel,dashed,lr,-] (h'');
    \end{tikzpicture}
  }
  \hfill
  \subfloat[\PRhcl]{
   \pgfmathsetseed{5}
   \begin{tikzpicture}[etltree,scale=.6]
      \node[unlabeled] (root) {};
      \path (root) ++(-1.4cm,2.8cm) node[labeled] (h) {$h$}
                   ++(-.7cm,1.4cm) node[labeled] (he) {$he$};
      \path (root) ++(.75cm,1.5cm) node[unlabeled] (h''1) {}
                   ++(.45cm,.9cm) node[unlabeled] (h''2) {}
                   ++(.45cm,.9cm) node[unlabeled] (h''3) {}
                   ++(.45cm,.9cm) node[labeled] (h') {$h'$};

      \draw[eventrel] (h) -- node {e} (he);
      \draw[eventrel] (h''1) -- (h''2);
      \draw[eventrel] (h''2) -- (h''3);
      \draw[eventrel] (h''3) -- (h');
      { [eventrel,decoration={random steps,segment length=1mm,amplitude=.5mm}]
        \draw[decorate] (root) -- (h); %node {h} (h);
        \draw[decorate] (root) -- (h''1); %node {h''} (h'');
      }
%      \draw[decorate,decoration={brace,mirror,raise=8pt}] (root) -- node[right,event,xshift=4mm,yshift=-2mm] {$h'$} (h');
      % \node[labeled,right,inner sep=2pt] at (h') {$h'$};
      % \node[labeled,left,inner sep=2pt] at (h) {$h$};
      % \node[labeled,left,inner sep=2pt] at (he) {$he$};

      \path (he) edge[accrel,lr,-] (h');
      \path (he) edge[accrel,dashed,lr,-] (h''3);
      \path (he) edge[accrel,dashed,lr,-] (h''2);
      \path (h) edge[accrel,dashed,lr,-] (h''1);
    \end{tikzpicture}
  }
  \hfill~
  \caption{\PRhcl inductively bridges the gap to \PRhc, but it imposes additional conditions on the intermediate states.
    However, in S5 the two conditions are equivalent.}
  \label{fig:bridging-hc-hcl}
\end{figure}

Intuitively it is thus clear that \PRhcl is at least as strong a condition as \PRhc,
as formalized in the next result.

\begin{lemma}
  \label{result:pr-hcl-implies-pr-hc}
  Any ETL frame that has \PRhcl also has \PRhc.
\end{lemma}
\begin{proof}
  A simple induction on the length of $h'$ in the definition of \PRhcl proves the claim.
\end{proof}

In \cref{sec:sub-s5-settings} we will see that \PRhcl is in general strictly stronger than \PRhc;
and in particular in \cref{fig:pr-ee-pr-hc} that it detects situations where (intuitively) information is lost,
for which \PRhc fails to do so.
However, in the context of S5, the notions are in fact equivalent.

\begin{proposition}
  \label{result:pr-ee-equiv-pr-hc-equiv-pr-hcl}
  An S5 ETL frame has \PR \tiff it has \PRhcl.
\end{proposition}
\begin{proof}
  Recall that we use \PR to refer to \PRhc and/or \PRee,
  since they are equivalent according to \cref{result:pr-ee-equiv-pr-hc}.
  \PRhcl implies \PRhc by \cref{result:pr-hcl-implies-pr-hc}.
  Vice versa, a simple induction on the length of $h'$ shows that \PRhc and \PRee imply \PRhcl.
\end{proof}

Thus, in the context of S5, \PRhcl is equivalent to the ``established'' notions and
we can use it to characterize perfect recall frames.

\begin{theorem}
\label{result:pr-hcl-definable}
  An ETL frame has \PRhcl \tiff it validates the following formula for each event $e$
  (recall that $\nextstep\phi$ abbreviates $\bigvee_{e'}\after{e'} \phi$):
  \begin{equation}
    \label{eq:axiom-pr-hcl}\tag{$\star$}
    \after e L p\implies L p \vee L\nextstep p \vee \after e L \nextstep p\mpunct.
  \end{equation}
  Together with the normal modal logic axioms and deduction rules,
  it is sound and complete with respect to the class of ETL frames with \PRhcl.
\end{theorem}
\begin{proof}
  This follows from the Correspondence and Completeness theorems by \citet{sahlqvist_completeness_1975}.
  In some more detail, note first that the formula as frame property is equivalent to
  \[
  K p \wedge K\nextstepall p \wedge \afterall e K \nextstepall p \implies \afterall e K p\mpunct,
  \]
  where $\nextstepall\phi$ abbreviates $\bigwedge_{e'}\afterall{e'}\phi$.
  This is a Sahlqvist formula, so as an axiom it is complete
  with respect to the class of frames it defines.
  To see what that class is, we start from the second-order formulation of this frame property
  (for any histories $h,h'$ and event $e'$ we write $h\leadsto_{e'}h'$ \tiff $he'=h'$,
  and we write $h\leadsto h'$ \tiff $h\leadsto_{e'}h'$ for some $e'$):
  \begin{align*}
    \forall P\forall h_1&\big[
    \begin{aligned}[t]
      &\forall h_2 (h_1\acc h_2\implies P h_2)\\
      &\wedge\forall h_2\forall h_3(h_1\acc h_2\leadsto h_3 \implies P h_3)\\
      &\wedge\forall h_2\forall h_3\forall h_4(h_1\leadsto_e h_2\acc h_3\leadsto h_4\implies P h_4)\big]
    \end{aligned}\\
    &\implies\forall h_2\forall h_3(h_1\leadsto_e h_2\acc h_3\implies P h_3)
  \end{align*}
  As \citeauthor{sahlqvist_completeness_1975} pointed out,
  since $P$ does not occur negated in the consequent,
  the minimal instantiation of $P$ satisfying the antecedent
  yields an equivalent first-order formula.
  This minimal instantiation can be read off as:
  \[
  P h := h_1\acc h\vee\exists h_2(h_1\acc h_2\leadsto h)\vee
  \exists h_2\exists h_3(h_1\leadsto_e h_2\acc h_3\leadsto h)\mpunct.
  \]
  Since it satisfies the antecedent,
  we are left with the instantiated consequent:
  \begin{align*}
    \forall h_1\forall h_2\forall h_3\big[
    &h_1\leadsto_e h_2\acc h_3\implies\\
    &h_1\acc h_3\\
    &\vee\exists h_4(h_1\acc h_4\leadsto h_3)\\
    &\vee\exists h_4\exists h_5(h_1\leadsto_e h_4\acc h_5\leadsto h_3)\big]\mpunct.
  \end{align*}
  Since there is at most one $e$-successor for any given history $h$ (namely $he$),
  we can replace $h_4$ in the last disjunct by $h_2$.
  It is then easy to see that this is equivalent to \PRhcl.
\end{proof}

\begin{corollary}
  \label{result:pr-definable}
  An S5 ETL frame has \PR \tiff it validates \eqref{eq:axiom-pr-hcl}.
\end{corollary}
\begin{proof}
  Immediate with \cref{result:pr-ee-equiv-pr-hc-equiv-pr-hcl,result:pr-hcl-definable}.
\end{proof}

We thus have an axiomatization of S5 ETL with perfect recall.
However, some of the results we used indeed depended on S5.
If we give up S5, we have to take a fresh look at certain issues,
and that is the topic of the next section.

\section{Sub-S5 settings}
\label{sec:sub-s5-settings}

To our knowledge, perfect recall has only been considered in the context of S5 in the literature,
a likely reason being that both communities that have studied the notion most
(interpreted systems and game theory)
virtually exclusively consider S5 settings.
However, it may make perfect sense, for example,
to say of a misinformed agent that he correctly remembers all information he has ever had,
even if that information itself is not correct.
In this section, we explore such settings of general ETL frames.

\subsection{Preliminaries}
\label{sec:preliminaries}

We stick with the symmetric-looking symbol~$\acc$ even when the relation is not necessarily symmetric.
Note that $[h]_\acc$ now does not necessarily contain $h$ itself anymore,
but we have the following fact.
\begin{fact}
  \label{result:trans-eucl-eq-class}
  If~$\acc$ is a transitive and Euclidean relation,
  then it is an equivalence relation on $[h]_\acc$ for any $h$
  (cf.~\cite[Theorem~3.3]{halpern_relationship_1991}).
  Consequently, for any $h,h'$ with $[h]_\acc\cap[h']_\acc\neq\emptyset$,
  we have $[h]_\acc=[h']_\acc$.
\end{fact}

\medskip

Note that the motivation and justifications for the definitions of perfect recall we gave
did not assume S5 knowledge.
Each of the notions captured a particular way of not losing information,
and they still make sense without S5.
We therefore take over the basic definitions without any change,
but take a new look at how they relate.

\subsection{Contrasting the notions}
\label{sec:contrasting-notions}

First note that \PRhcl still implies \PRhc,
since \cref{result:pr-hcl-implies-pr-hc} did not assume S5.
However, without any assumptions about the frames,
none of the other mutual implications among \PRhc, \PRhcl and \PRee remain.
This is witnessed by \cref{fig:pr-ee-pr-hc},
illustrating that epistemic experience and history consistency
reflect two different ways of remembering past information.

\begin{figure}
  \centering
  \subfloat[$\PRee$, but not $\PRhc$ nor $\PRhcl$]{
    \label{fig:pr-ee-pr-hc:ee-not-hc}
    \begin{tikzpicture}[etltree]
      \node[unlabeled] (root) {}
        child foreach \c in {1,2} {
          node {} child { node {} edge from parent node {$e_3$} }
          edge from parent node {$e_\c$}
        };
      \draw[accrel] (root-2) edge[rl] (root-1)
            (root-1-1) edge[lr,<->] (root-2-1)
            (root) edge[loop below] ()
            (root-1) edge[loop left] ()
            (root-1-1) edge[loop left] ()
            (root-2-1) edge[loop right] ();
    \end{tikzpicture}
  }
  \hfill
  \subfloat[$\PRhc$ and $\PRhcl$, but not $\PRee$]{
    \label{fig:pr-ee-pr-hc:hc-not-ee-1}
    \begin{tikzpicture}[etltree]
      \node[unlabeled] (root) {}
        child foreach \c in {1,2,3} {
          node {}
          edge from parent node {$e_\c$}
        };
      \draw[accrel] (root-1) edge[lr] (root-2) 
            (root-2) edge[lr] (root-3)
            (root) edge[loop below] ()
            (root-3) edge[loop above] ();
    \end{tikzpicture}
  }
  \hfill
  \subfloat[$\PRhc$ and $\PRhcl$, but not $\PRee$]{
    \label{fig:pr-ee-pr-hc:hc-not-ee-2}
%    ~
    \begin{tikzpicture}[etltree]
      \node[unlabeled] (root) {}
        child {
          node {}
          edge from parent node {$e_1$}
        }
        child {
          node {}
          child {
            node {}
            edge from parent node {$e_3$}
          }
          edge from parent node {$e_2$}
        };
      \draw[accrel] (root-1) edge[lr] (root-2-1) 
            (root) edge[lr,in=120] (root-2)
            (root) edge[loop below] ()
            (root-2) edge[loop right] ()
            (root-2-1) edge[loop right] ();
    \end{tikzpicture}
    ~
  }
  \hfill
  \subfloat[$\PRhc$, but not $\PRhcl$ nor $\PRee$]{
    \label{fig:pr-ee-pr-hc:hc-not-hcl}
%    ~
    \begin{tikzpicture}[etltree]
      \node[unlabeled] (root) {}
        child {
          node {}
          edge from parent node {$e_1$}
        }
        child {
          node {}
          child {
            node {}
            edge from parent node {$e_3$}
          }
          edge from parent node {$e_2$}
        };
      \draw[accrel] (root-1) edge[lr] (root-2-1) 
            (root) edge[loop below] ()
            (root-2) edge[loop right] ()
            (root-2-1) edge[loop right] ();
    \end{tikzpicture}
%    ~
  }
  \caption{\PRhc, \PRhcl, and \PRee compared on general ETL frames,
    gray arrows depicting the accessibilities.
    \PRhc is violated in \subref{fig:pr-ee-pr-hc:ee-not-hc}
    since $e_1e_3\acc e_2e_3$ but there is no prefix of $e_2e_3$ that is accessible from $e_1$.
    \PRee is violated in \subref{fig:pr-ee-pr-hc:hc-not-ee-1}
    since $\EE(e_1)\not\approx\EE(e_2)$,
    and in \subref{fig:pr-ee-pr-hc:hc-not-ee-2} and \subref{fig:pr-ee-pr-hc:hc-not-hcl}
    since $\EE(e_1)\not\approx\EE(e_2e_3)$.
    \PRhcl is violated in \subref{fig:pr-ee-pr-hc:hc-not-hcl}
    since $e_1\acc e_2e_3$ but $\epsilon\not\acc e_2e_3$ and $\epsilon\not\acc e_2$ and $e_1\not\acc e_2$.
}
  \label{fig:pr-ee-pr-hc}
\end{figure}
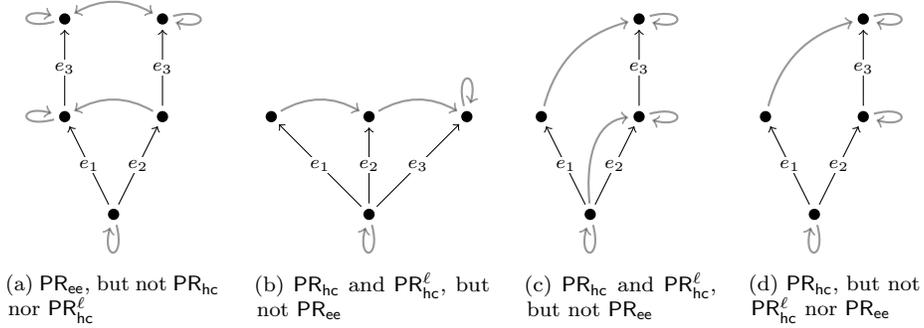

\begin{itemize}
  \item[\subref{fig:pr-ee-pr-hc:ee-not-hc}]
    An agent that only has perfect recall with respect to epistemic experience
    may at some point be certain that a particular history can be excluded,
    but later on ``forget'' this piece of information.
    In particular, at $e_1e_3$ the agent considers $e_2e_3$ possible,
    even though he never considered $e_2$ possible.
  \item[\subref{fig:pr-ee-pr-hc:hc-not-ee-1} and \subref{fig:pr-ee-pr-hc:hc-not-ee-2}]
    An agent that only has perfect recall with respect to history consistency
    may at some state consider another state possible,
    although in that other state his epistemic experience would have been different.
    For example, an agent at $e_1$ in \subref{fig:pr-ee-pr-hc:hc-not-ee-2}
    is (mistakenly) certain that he is at $e_2e_3$,
    even though in that state his previous information set would have been $\{e_2\}$,
    which contradicts his actual epistemic experience.
  \item[\subref{fig:pr-ee-pr-hc:hc-not-hcl}]
    The intuition is similar to the previous case,
    but here we see that \PRhcl is more fine-grained than \PRhc.
    An agent at $e_1$ thinks he is at $e_2e_3$,
    even though he never considered $e_2$ possible.
    He thus thinks himself at the endpoint of a history whose unfolding he deemed impossible.
    \PRhc grants this agent the label of perfect recall,
    while \PRhcl denies it.
\end{itemize}

Note that, while these phenomena reflect some kind ``\emph{forgetting}'',
they do not at first glance constitute a coherent, rational method of \emph{belief revision}.
A full-fledged doxastic logic is needed in order to really model agents
that reconsider their previous assessments and deal with ``unwanted'' memories properly.

\bigskip

The following straightforward result enables us to identify the settings
in which the different notions of perfect recall can be meaningfully compared.

\begin{proposition}
\label{result:pr-ee-implies-45}
  Any ETL frame that has \PRee is transitive and Euclidean.
\end{proposition}
\begin{proof}
  This is obvious from \cref{dfn:pr-ee}.
  For example, for any three histories $h,h',h''$, if $h\acc h'$ and $h\acc h''$,
  then \PRee implies that $[h]_\acc=[h']_\acc=[h'']_\acc$.
  Since $h',h''\in [h]_\acc$, we also get $h'\in [h'']_\acc$ and $h''\in [h']_\acc$.
\end{proof}

This result is not very surprising, given that \PRee
requires of an agent to be able to assess his own epistemic experience---%
including at the current state.
It implies that any reflexive ETL frame with \PRee is already an S5 ETL frame.
On the level of agents, an agent who has correct beliefs and \PRee
has in fact already full (S5) knowledge.
However, perfect recall does not \emph{require} the agent to have correct beliefs
(in fact, perfect recall by itself is compatible with believing falsum).
For example, KD45 is a common sub-S5 setting in which perfect recall is a meaningful notion.

Given that transitivity and Euclideanness are inherent to \PRee,
we continue our comparison within the corresponding class of frames.
In the following, we use \dfn{introspective} to mean ``transitive and Euclidean''.

\subsection{Characterizing \PRee locally}
\label{sec:intr-sett}

Given the fact that \PRhcl no longer characterizes \PRee
and the original definition of \PRee is somewhat unwieldy,
it can be useful to have a local condition on histories and accessibilities that corresponds to it.
It turns out that we can re-use \PRhcl by slightly modifying the frame in question.

For a frame \F with accessibility relation~$\acc$, we will use $\FSfive$ and $\accSfive$ to denote the S5~closure.
We need the following small technical condition:
We say that a frame satisfies \dfnless{persistent insanity}
if, whenever $[h]_\acc=\emptyset$ and $h\preceq h'$, then $[h']_\acc=\emptyset$.
Intuitively, once a corresponding agent has inconsistent beliefs,
he will remain in that pitiful condition forever.

\begin{proposition}
  \label{result:introsp-pr-ee-iff-closure-pr-hcl}
  An introspective ETL frame satisfying persistent insanity has \PRee \tiff its S5 closure has \PRhcl.%
  \footnote{\label{fn:s5-closure-is-local}
    Note that this is indeed a local condition:
    On introspective frames the S5 closure is the symmetric and reflexive closure,
    without any need of iterating through the accessibility relation (cf.~\cref{result:trans-eucl-eq-class}).}
\end{proposition}
\begin{proof}
  Due to \cref{result:pr-ee-equiv-pr-hc-equiv-pr-hcl},
  \PRee and \PRhcl are equivalent on the S5 closure.
  We can thus prove the claim by showing that
  \PRee is invariant under taking this closure.

  To see that this is indeed the case, take any pair $h,h'$ of histories,
  the accessibility relation~$\acc$ of an introspective frame satisfying persistent insanity,
  and its S5 closure~$\accSfive$.
  With \cref{result:trans-eucl-eq-class}, it is easy to see that,
  as long as $[h]_\acc\neq\emptyset$, we have
  $[h]_\acc=[h']_\acc$ \tiff $[h]_{\accSfive}=[h']_{\accSfive}$.
  Inductively it follows that the equivalence of epistemic experiences
  is invariant under taking the S5 closure as long as $[h]_\acc\neq\emptyset$;
  persistent insanity ensures that \PRee is also satisfied for any $h'$
  extending an $h$ with $[h]_\acc=\emptyset$.
\end{proof}

To see that persistent insanity is indeed needed for this result,
consider \cref{fig:pr-ee-pr-hc}\subref{fig:pr-ee-pr-hc:ee-not-hc}
with the accessibilities $e_2\acc e_1$ and $e_1\acc e_1$ removed.
The resulting frame still has \PRee, but its S5 closure does not have \PRhcl.

\begin{corollary}
  A KD45 ETL frame has \PRee \tiff its S5 closure has \PRhcl.
\end{corollary}
\begin{proof}
  Immediate, since KD45 frames are introspective and vacuously satisfy persistent insanity.
\end{proof}

\begin{numberedremark}
  \label{result:pr-hcl-implies-pers-ins}
  Note that \PRhcl implies persistent insanity,
  so \cref{result:introsp-pr-ee-iff-closure-pr-hcl} applies to all introspective frames with \PRhcl.
\end{numberedremark}

As witnessed by \cref{fig:pr-ee-pr-hc} and by \cref{fig:4-5-pr-hc-not-pr-ee} later on,
the notions \PRee and \PRhcl are incomparable in the sense that
each one is stronger than the other one under certain circumstances.
\Cref{result:introsp-pr-ee-iff-closure-pr-hcl} gives an insight as to why this is so:
By applying the \PRhcl condition to the S5 closure of a frame,
on the one hand the antecedent in this condition becomes more permissive,
but on the other hand so does the consequent.
Thus, the condition gets both strengthened and weakened.

\bigskip

Now that we have contrasted our basic notions of perfect recall
and provided and discussed separate local characterizations,
we proceed to characterize the combination of the notions.
We use \PR to denote the combination of perfect recall notions,
\PRee plus \PRhcl (and thus \PRhc),
describing perfect-recall agents that can reason both
about their epistemic experience and about history consistency.

As mentioned earlier,
our considerations so far hold both for ETL trees and ETL forests.
Now, however, the distinction becomes important.
We start by focusing on trees.

\subsection{Characterizing \PR on trees}

It turns out that on introspective trees, \PRhcl captures both definitions of perfect recall,
much like it (and \PRhc) did on S5 frames.
This allows us to define and axiomatize \PR on introspective trees,
reusing the results we obtained in \cref{sec:char-pras-etl}.

\begin{theorem}
  \label{result:pr-hcl-equiv-pr}
  An introspective ETL tree has \PR \tiff it has \PRhcl.
\end{theorem}

Note that this result is not in contradiction with the examples in \cref{fig:pr-ee-pr-hc},
since \subref{fig:pr-ee-pr-hc:hc-not-ee-1} is not transitive and \subref{fig:pr-ee-pr-hc:hc-not-ee-2} is not Euclidean.
For the proof, we need the following auxiliary results.

\begin{observation}
  \label{result:acc-along-history}
  For any introspective ETL frame with \PRhcl and
  histories $h_1,h_2$ with $h_1\preceq h_2$ and $h_1\acc h_2$,
  for each $h_1'\preceq h_1$ there is $h_2'\preceq h_2$ such that $h_1'\acc h_2'$.
\end{observation}
\begin{proof}
  The claim can be shown with a simple induction on $h_1$,
  using \cref{result:pr-hcl-implies-pr-hc} and transitivity of $\preceq$.
\end{proof}

\begin{lemma}%[Flashlight]
  \label{result:flashlight}
  For any introspective ETL frame with \PRhcl and
  histories $h_1$ and $h_2'\preceq h\preceq h_2$, if $h_1\acc h_2$ and $h_1\acc h_2'$
  then $h_1\accSfive h$.
\end{lemma}
\begin{proof}
  With Euclideanness, we obtain $h_2\acc h_2'$.
  Let $h_2''$ be the shortest prefix of $h_2$ such that $h_2\acc h_2''$ (note that $h_2''\preceq h_2'\preceq h$).
  \Cref{result:acc-along-history} implies that there must be $h'\preceq h_2''$ with $h\acc h'$.
  Another application of \cref{result:acc-along-history} then yields that
  there is $h''\preceq h'$ such that $h_2''\acc h''$,
  and by transitivity we get $h_2\acc h''$.
  Now if $h'\prec h_2''$ then $h''\prec h_2''$,
  contradicting that $h_2''$ is the shortest prefix accessible from $h_2$.
  So $h'=h_2''$, thus $h\acc h_2''$.
  Since $h_1\acc h_2\acc h_2''$, we obtain the claim.
\end{proof}

\begin{lemma}
  \label{result:pr-hcl-preserved-under-s5-closure}
  If an introspective ETL tree has \PRhcl, then so does its S5 closure.
\end{lemma}
\begin{proof}
  Take any introspective tree \F with \PRhcl.
  To see that its S5 closure~$\FSfive$ also has \PRhcl, let $h,e,h'$ be such that $he\accSfive h'$.
  We need to show that $\accSfive$ satisfies one of the three conditions in the definition of \PRhcl.

  Since~$\accSfive$ is the symmetric and reflexive closure of~$\acc$
  (cf.~footnote~\ref{fn:s5-closure-is-local}),
  we have either of these cases:
  \begin{itemize}
  \item $he=h'$.
    Since~$h\accSfive h$, condition~(\ref{dfn:pr-hcl:2}) in the definition of \PRhcl obtains.
  \item $he\acc h'$.
    Since \F has \PRhcl, $\acc$ satisfies one of the three conditions of \PRhcl, thus so does $\accSfive$.
  \item $h'\acc he$.
    If $h'\neq\epsilon$ then the same argument as in the previous case applies.
    Otherwise, $h'=\epsilon\acc he$.
    Euclideanness of $\acc$ yields $he\acc he$,
    and since $\acc$ satisfies \PRhcl, we have either of these three cases:
    \begin{enumerate}[(i)]
    \item $h\acc he$.
      Since $h'\acc he$, we get $h\accSfive h'$,
      so $\accSfive$ satisfies condition~(\ref{dfn:pr-hcl:1}) of \PRhcl.
    \item $h\acc h$.
      Since $h'=\epsilon\preceq h$,
      \cref{result:acc-along-history} yields that there is $h''\preceq h$
      such that $h'\acc h''$.
      Now we have $h'\acc he$, $h'\acc h''$ and $h''\preceq h\preceq he$,
      so \cref{result:flashlight} applies and yields $h'\accSfive h$.
      Symmetry of~$\accSfive$ yields $h\accSfive h'$,
      so $\accSfive$ satisfies condition~(\ref{dfn:pr-hcl:1}) of \PRhcl.
    \item $he\acc h$.
      Together with $h'\acc he$ we get $h'\accSfive h$.
      Symmetry of~$\accSfive$ again yields $h\accSfive h'$,
      condition~(\ref{dfn:pr-hcl:1}) of \PRhcl. \qedhere
    \end{enumerate}
  \end{itemize}
\end{proof}

We can now straightforwardly prove the stated result.
\begin{proof}[Proof of \cref{result:pr-hcl-equiv-pr}]
  \PR implies \PRhcl by definition.
  To see that the reverse direction holds,
  take any introspective ETL tree \F that has \PRhcl.
  Due to \cref{result:pr-hcl-preserved-under-s5-closure},
  its S5 closure also has \PRhcl,
  and with \cref{result:introsp-pr-ee-iff-closure-pr-hcl,result:pr-hcl-implies-pers-ins}
  it follows that \F also has \PRee.
\end{proof}

\medskip

% On introspective trees
% we thus arrive at \PRhcl as the single fundamental definition of perfect recall.
% As shown, it implies all notions of perfect recall that we have studied,
% which is not the case for any of the other notions by themselves.
% Regarding the requirements of transitivity and Euclideanness,
% it is questionable whether it makes sense
% to consider perfect recall in the absence of introspection.
% As shown in \cref{result:pr-ee-implies-45}, \PRee inherently postulates it;
% intuitively, too, we may want to ask of a ``perfectly recalling'' agent
% that he can also perfectly assess his current information state,
% that is, that he has positive and negative introspection.

Since the proof of \cref{result:pr-hcl-definable} did not use S5,
we immediately obtain an axiomatization of ETL
on introspective trees with perfect recall.
Further, it is easy to see that on synchronous trees, \PRs and \PRhcl are still equivalent,
so we also have an axiomatization of ETL on synchronous introspective trees with perfect recall.

\subsection{Characterizing \PR on forests}
\label{sec:char-pr-forests}

\Cref{result:pr-hcl-equiv-pr} does not apply to forests,
as witnessed by \cref{fig:4-5-pr-hc-not-pr-ee}\subref{fig:4-5-pr-hc-not-pr-ee:not-pr-ee}.
Before we look at how to characterize \PR here,
we note that, unlike on S5 forests, defining \PR on introspective forests generally is impossible in ETL
(the same holds for \PRee).

\begin{proposition}
  \PR is not modally definable on introspective ETL forests
  (and thus not on general ETL forests either).
\end{proposition}
\begin{proof}
  $\F'$ in \cref{fig:4-5-pr-hc-not-pr-ee} has \PR,
  while its bounded morphic image \F does not.
  Since modally definable properties are closed under bounded morphic images
  (cf.~\cite{blackburn_modal_2001}),
  the claim follows.
\end{proof}

\begin{figure}
  \centering
  \hfill
  \subfloat[Forest \F (two trees)]{
    \label{fig:4-5-pr-hc-not-pr-ee:not-pr-ee}
    \begin{tikzpicture}[etltree,remember picture]
      \node[unlabeled] (root1) {}
        child {
          node {}
          edge from parent node {$e_1$}
        };
      \node[unlabeled,right of=root1] (root2) {};
      \draw[accrel] (root1) edge[loop left] () 
            (root1-1) edge[loop left] ()
            (root2) edge[rl] (root1-1);
      \useasboundingbox (current bounding box.north east)+(1,.7) rectangle (-1,-.7);
    \end{tikzpicture}
  }
  \hfill
  \subfloat[Forest $\F'$ (three trees)]{
    \hfill
    \label{fig:4-5-pr-hc-not-pr-ee:pr-ee}
    \begin{tikzpicture}[etltree,remember picture]
      \node[unlabeled] (root1') {}
        child {
          node {}
          edge from parent node {$e_1$}
        };
      \node[unlabeled,right of=root1'] (root2') {};
      \node[unlabeled,right of=root2'] (root3') {};
      \draw[accrel]
        (root1') edge[loop left] () 
        (root1'-1) edge[loop left] ()
        (root2') edge[lr] (root3')
        (root3') edge[loop right] ();

      {[overlay,mapping]
        \draw
          (root1') edge[bend left] (root1)
          (root2') edge[bend left] (root2)
          (root1'-1) edge[bend right=40] (root1-1)
          (root3') edge[bend right=70] (root1-1);
      }
      \useasboundingbox (current bounding box.north)+(0,.6) rectangle (0,-.6);
    \end{tikzpicture}
    \hfill
  }
  \hfill~
  \caption{\subref{fig:4-5-pr-hc-not-pr-ee:not-pr-ee} \F is an introspective forest that has \PRhcl, but not \PRee.
    \subref{fig:4-5-pr-hc-not-pr-ee:pr-ee} $\F'$ has \PRhcl and \PRee, and \F is its bounded morphic image
    via the bounded morphism depicted with dashed arrows.
    So \PR is not modally definable on forests.}
  \label{fig:4-5-pr-hc-not-pr-ee}
\end{figure}
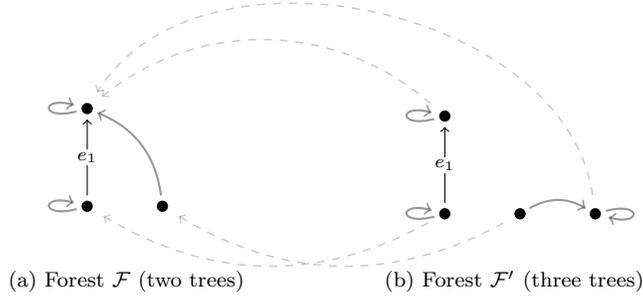

% We first obtain a straightforward characterization of \PR.
% \begin{proposition}
%   Any introspective frame has \PR iff both it and its S5 closure have \PRhcl.
% \end{proposition}
% \begin{proof}
%   Follows directly from \cref{result:introsp-pr-ee-iff-closure-pr-hcl,result:pr-hcl-implies-pers-ins}.
% \end{proof}

% The following result simplifies this condition slightly.
% \begin{theorem}
%   Any introspective ETL forest \F has \PR iff both \F and $\F^{-1}$ have \PRhcl.
% \end{theorem}
% \begin{proof}
%   Given\cref{result:introsp-pr-ee-iff-closure-pr-hcl,result:pr-hcl-implies-pers-ins},
%   the only nontrivial part left to prove is that if \F has \PR then $\F^{-1}$ has \PRhcl.

% false, consider this \F:
%     \begin{tikzpicture}[etltree]
%       \node[unlabeled] (root1) {}
%         child {
%           node {}
%           edge from parent node {$e_1$}
%         };
%       \draw[accrel]
%             (root1-1) edge[loop left] ()
%             (root1) edge[rl] (root1-1);
%     \end{tikzpicture}
% \end{proof}

From \cref{result:introsp-pr-ee-iff-closure-pr-hcl,result:pr-hcl-implies-pers-ins}
it is clear that 
any introspective frame has \PR \tiff both it and its S5 closure have \PRhcl.
However, with an additional slight restriction, we can continue to use \PRhcl to characterize \PR.
From \cref{fig:4-5-pr-hc-not-pr-ee}, it is intuitively clear that
accessibilities from some root to a later state in some (different) tree are problematic:
In such cases, \PRhcl is vacuously satisfied, while \PRee may not hold.

To fix this, call an ETL forest \F \dfnless{initially synchronous} if,
for any two roots~$\epsilon,\epsilon'$ and history~$h$ with $\epsilon\preceq h$ and $\epsilon'\acc h$,
we also have $\epsilon'\acc\epsilon$.
That is, the agent at least considers it possible that indeed no time has passed initially,
although he may immediately lose synchronicity and also consider later states possible.
We then get the following.
\begin{lemma}
  \label{result:pr-hcl-preserved-under-s5-closure-on-init-sync-forests}
  If an introspective and initially synchronous ETL forest has \PRhcl, then so does its S5 closure.
\end{lemma}
\begin{proof}
  The proof is analogous to that of \cref{result:pr-hcl-preserved-under-s5-closure},
  with one additional observation:
  If $\epsilon'\acc h$ for some history~$h$ with root~$\epsilon$,
  then initial synchronicity yields $\epsilon'\acc\epsilon$.
  Euclideanness then yields $\epsilon\acc h$,
  and with \cref{result:pr-hcl-equiv-pr} it follows that $\EE(\epsilon)\approx\EE(h)$.
  Due to \cref{result:trans-eucl-eq-class}, we also have $\EE(\epsilon')\approx\EE(\epsilon)$,
  so $\EE(\epsilon')\approx\EE(h)$ by transitivity of~$\approx$.
\end{proof}

\begin{theorem}
  An introspective and initially synchronous ETL forest has \PR \tiff it has \PRhcl.
\end{theorem}
\begin{proof}
  Analogously to \cref{result:pr-hcl-equiv-pr}, this follows from
  \cref{result:pr-hcl-preserved-under-s5-closure-on-init-sync-forests,result:introsp-pr-ee-iff-closure-pr-hcl,result:pr-hcl-implies-pers-ins}.
\end{proof}

\section{Conclusions}
\label{sec:conclusions}

We discussed two different ways of ``not losing information'', that is,
accessing and reasoning with one's memories.
The first one has been the fundamental definition in the literature on perfect recall.
It assumes that a perfect-recall agent can use differences in past epistemic states
in order to distinguish present states.
The second one is a consistency condition on the histories considered possible.
While it has been used in the literature as technical condition,
we provided it with its own motivation.

The two notions have previously been studied in~S5, where they coincide,
and in logics with long-range temporal operators.
We gave a novel characterization and axiomatization in ETL,
using only next-step temporal operators.

We then dropped the assumption of~S5
and noticed that the notions no longer coincide.
Since they capture two independently motivated ways of reasoning with memories,
we examined and characterized them individually as well as jointly.

\medskip

Given that the two notions use different aspects of ETL models,
some discussion is needed concerning the access that we assume an agent to have.

It is a general issue in modeling agents
to what extent the model faithfully represents an agent's internal workings,
and to what extent it represents the modeler's external perspective.
What we mean if we say that an agent ``does not lose information'', of course,
depends on what information we ascribe to him in the first place.
ETL is agnostic as to whether the agent has direct access to the semantic structures constituting a model
or whether they are just a representation for the modeler,
and whether the logic language is supposed to reflect the agent's ``mentalese''
or whether it is just a way for the modeler to talk \emph{about} the agent.
Depending on the intended interpretation,
one may exclude or include certain features in what is considered the agent's information,
and one may accept or reject certain methods for the agent to access and reason with his memories.

Since ETL does not specify these issues,
we simply examined what can be said \emph{if} the agent has access to certain aspects of the model.
Outside of S5, where the notions do not coincide,
it depends on the modeled situation
which definition of perfect recall is the right one.

%todo{describe that from K[]p and something happens cannot be deduced K p, so the two accesses are indeed different}
\medskip

An interesting question for further research is
whether there are additional aspects of reasoning with memories,
which might be conflated in S5 with the ones we discussed, but distinct in other settings.

\medskip

The general goal with these considerations is
to help improve our understanding of
the assumptions implicit in the framework or explicitly made by the modeler when modeling agents.
Along the lines of the inspiring work by~\citet{van_benthem_merging_2009},
we hope to obtain more fine-grained insights in more general settings.

\section*{Acknowledgments}
\label{sec:acknowledgments}

This work came out of discussions with Benedikt Löwe and Cédric Dégremont.
Thanks also to Krzysztof Apt, Can Başkent, and Johan van Benthem for comments.
% Do the two formulations differ qua complexity properties?
% E.g., does yours, in the model classes where it differs, escape
% from the fateful Halpern & Vardi Tiling encoding that leads
% to the epistemic temporal logic becoming PI-1-1-complete?
%\cite{bonanno_memory_2004} with turns and past operator

\bibliographystyle{plainnat}
\bibliography{all}

\end{document}